\documentclass[centertags,12pt]{amsart}
\usepackage{amssymb}
\usepackage{graphics}
\usepackage{graphicx}
\usepackage{latexsym}
\usepackage{color}

\makeatletter
\renewcommand{\subsection}{\@startsection
{subsection}{2}{0mm}{\baselineskip}{-0.25cm}
{\normalfont\normalsize\em}}
\makeatother

\newtheorem{theorem}{Theorem}[section]
\newtheorem{proposition}[theorem]{Proposition}

\newtheorem{lemma}[theorem]{Lemma}

{\theoremstyle{definition}

\newtheorem{example}[theorem]{Example}

}
\theoremstyle{remark}

\def\codim{\ensuremath{r}}
\def\MDSr{\ensuremath{t}}
\def\dry{\ensuremath{\delta}}

\begin{document}
\title{Wet paper codes and the dual distance in steganography}
\author{C. Munuera and M. Barbier}
\thanks{
C. Munuera is with the Department of Applied Mathematics, University of Valladolid,
Avda Salamanca SN, 47014 Valladolid, Castilla, Spain.
M. Barbier is with the Computer Science laboratory,  \'Ecole Polytechnique, 91 128 Palaiseau CEDEX,
France. INRIA Saclay, \^Ile de France} 

\sloppypar

\begin{abstract}
  In 1998 Crandall introduced a method based on coding theory to
  secretly embed a message in a digital support such as an
  image. Later Fridrich \textit{et al.} improved this method to
  minimize the distortion introduced by the embedding; a process called
  \textit{wet paper}. However, as previously emphasized in the literature, this
  method can fail during the embedding step. Here we find 
  sufficient and necessary conditions to guarantee a successful embedding by studying
  the dual distance of a linear code. Since these results
  are essentially of combinatorial nature, they can be generalized
  to systematic codes,  a large family containing all linear
  codes. We also compute the exact number of solutions and point out 
  the relationship between wet paper codes and orthogonal arrays. 
\end{abstract}

\thanks{This work was supported in part by Junta de CyL under grant VA065A07 and by Spanish Ministry for Science and Technology under grants MTM2007-66842-C02-01 and  MTM 2007-64704.}

\maketitle


\section{Introduction}

Steganography is the science of transmitting messages in secret, so
that no one other than the sender and receiver may detect the
existence of hidden data. It is realized by embedding the information
into innocuous cover objects, as digital images. To carry out this
process, the sender first extracts a sequence $c_1,\dots,c_n$, of $n$
bits from the image, e.g. the least significant bits of $n$ pixels
gray values. The {\em cover vector} $\mathbf{c}=(c_1,\dots,c_n)$
is modified according to a certain algorithm for storing a {\em secret
message} $m_1,\dots,m_\codim$. Then $c_1,\dots,c_n$ are
replaced by modified $x_1,\dots,x_n$ in the {\em cover image} which is
sent through the channel. By using the modified vector $\mathbf{x}$
the receiver is able to recover the hidden information. The embedding
and recovering algorithms form the {\em steganographic scheme} of this
system. Formally,
a {\em steganographic scheme} (or {\em stegoscheme}) $\mathcal{S}$ of
{\em type} $[n,r]$ over the binary alphabet $\mathbb{F}_2$ is a pair
of functions $({\rm emb},{\rm rec})$.
By using the {\em embedding function} ${\rm
  emb}:\mathbb{F}_2^n\times\mathbb{F}_2^\codim\rightarrow\mathbb{F}_2^n$
the secret message  $\mathbf{m}\in\mathbb{F}_2^\codim$ is hidden in the
cover vector $\mathbf{c}\in\mathbb{F}_2^n$ as $\mathbf{x}={\rm
  emb}(\mathbf{c},\mathbf{m})$  and subsequently recovered by the
receiver with the {\em recovering function} ${\rm
  rec}:\mathbb{F}_2^n\rightarrow \mathbb{F}_2^\codim$ as ${\rm
  rec}(\mathbf{x})$ whenever these functions verify that
${\rm rec}({\rm emb}(\mathbf{c},\mathbf{m})) =\mathbf{m}$ for all
$\mathbf{c}\in\mathbb{F}_2^n$ and $\mathbf{m}\in\mathbb{F}_2^\codim$.

Steganographic schemes are closely related to error correcting
codes. Given a $[n,\codim]$ stego\-scheme $\mathcal{S}$,  for each
$\mathbf{m}\in\mathbb{F}_2^\codim$ we consider the code
$\mathcal{C}_{\mathbf{m}}=\{ \mathbf{x}\in\mathbb{F}_2^n : {\rm
  rec}(\mathbf{x})=\mathbf{m}\}$. Then the family $\{
\mathcal{C}_{\mathbf{m}} : \mathbf{m}\in\mathbb{F}_2^\codim  \}$ gives a
partition on $\mathbb{F}_2^n$ and for all
$\mathbf{m}\in\mathbb{F}_2^\codim$ the mapping
${\rm dec}_{\mathbf{m}}:\mathbb{F}_2^n\rightarrow{\mathcal
  C}_{\mathbf{m}}$ defined by  
${\rm dec}_{\mathbf{m}}(\mathbf{c})={\rm emb}(\mathbf{c},\mathbf{m})$
is a decoding map for the code $\mathcal{C}_{\mathbf{m}}$. Conversely
let $\{\mathcal{C}_\mathbf{m} : \mathbf{m}\in\mathbb{F}_2^\codim\}$ be a partition of $\mathbb{F}_2^n$ and for each
$\mathbf{m}\in\mathbb{F}_2^\codim$ let ${\rm dec}_{\mathbf{m}}$ be a
minimum distance decoding map of $\mathcal{C}_\mathbf{m}$. Consider ${\rm
  emb}:\mathbb{F}_2^n\times\mathbb{F}_2^\codim\rightarrow\mathbb{F}_2^n$ and
${\rm rec}:\mathbb{F}_2^n\rightarrow\mathbb{F}_2^\codim$ defined by 
${\rm emb}(\mathbf{c},\mathbf{m})={\rm dec}_{\mathbf{m}}(\mathbf{c})$
and ${\rm rec}(\mathbf{x})=\mathbf{m}$ if $\mathbf{x}\in
\mathcal{C}_\mathbf{m}$. Then $\mathcal{S}=({\rm emb, rec})$ is a
$[n,\codim]$ stegoscheme. 
As a consequence,  the following objects are equivalent 
\begin{itemize}
\item a $[n,\codim]$ stegoscheme $\mathcal{S}=({\rm emb, rec})$;
\item a family $\{(\mathcal{C}_\mathbf{m},{\rm dec}_{\mathbf{m}}) :
  \mathbf{m}\in\mathbb{F}_2^\codim\}$ where $\{\mathcal{C}_\mathbf{m} :
  \mathbf{m}\in\mathbb{F}_2^\codim\}$ gives a partition of $\mathbb{F}_2^n$
  and  for every $\mathbf{m}$, ${\rm dec}_{\mathbf{m}}$ is a minimum
  distance decoding map for $\mathcal{C}_{\mathbf{m}}$.
\end{itemize}

Since all vectors $\mathbf{c}$ and $\mathbf{m}$ are in principle equiprobable, it
is desirable that all codes $\mathcal{C}_{\mathbf{m}}$ have the same
cardinality.
The above equivalence has been extensively exploited to make
stegoschemes that minimize the embedding distortion caused in the
cover. Crandall \cite{crandall} proposed the use of linear codes
$\mathcal{C}$ and the partition of $\mathbb{F}_2^n$ into cosets \mbox{$\{
\mathbf{x}+\mathcal{C} : \mathbf{x}\in \mathbb{F}_2^n\}$}.
The obtained method is currently known as {\em matrix encoding}.  If
$H$ is a parity check matrix for $\mathcal{C}$ and dec is syndrome
decoding (see \cite{vLint} as a general reference for all facts
concerning error correcting codes), the obtained embedding and
recovering maps are  
\mbox{${\rm emb}(\mathbf{c},\mathbf{m})= \mathbf{c}-{\rm
    cl}(\mathbf{c}H^T-\mathbf{m})$} and ${\rm
  rec}(\mathbf{x})=\mathbf{x}H^T$, where ${\rm cl}(\mathbf{z})$
denotes the {\em leader} of the coset whose syndrome is
$\mathbf{z}$, that is the element of smallest Hamming weight whose
syndrome is $\mathbf{z}$ .
Matrix encoding has proved to be very efficient to
minimize the embedding distortion in the cover, see
\cite{covering,fontaine+galand,schonfeld,SW2,zhang}.

To reduce the chance of being detected by third parties, the
changeable pixels in the cover image should be selected according to
the characteristics of the image and the message to hide. In this case the
recovering of the hidden data is more difficult, since the receiver
does not know what pixels store information. Wet paper codes are
designed to lock some components of the cover vector, preventing its
modification in the embedding process.
Mathematically wet paper codes can be  explained as follows: 
imagine we want to embed a message
$\mathbf{m}=(m_1,\dots,m_\codim)\in\mathbb{F}_2^\codim$ into a cover
vector  $\mathbf{c}=(c_1,\dots,c_n)\in\mathbb{F}_2^n$. However, not
all of coordinates of $\mathbf{c}$ can be used for hiding information:
there is a set $\mathcal{D}\subseteq \{ 1,\dots,n\}$ of $\dry\ge
\codim$ {\em dry} coordinates that may be {\em freely} modified by the
sender, while the other $\ell=n-\dry$ coordinates are {\em wet} (or
{\em locked}) and can not be altered during the embedding process. Let
$\mathcal{W}=\{1,\dots,n\}\setminus \mathcal{D}$. The sets
$\mathcal{D}$ and $\mathcal{W}$ are known to the sender but not to the
receiver.
Using the matrix encoding method we set $\mbox{\rm
  emb}(\mathbf{c},\mathbf{m})=\mathbf{x}\in\mathbb{F}_2^n$ with 
\[
[S]:\left\{ \begin{array}{rl}
\mathbf{x}H^{T}&=\mathbf{m},\\
x_i&= c_i \; \mbox{ if $i\in\mathcal{W}$}
\end{array} \right.
\]
where $H$ is a $\codim\times n$ matrix of full rank $\codim$. Locking
positions minimize the possibility of detection during
transmission but also generates a technical problem,
since it is not guaranteed the existence of solutions for $[S]$.
A natural question is to ask for the minimum number of dry coordinates
(or equivalently, the maximum number of locked coordinates) necessary
(respectively allowed) to make possible this process. 
We define the {\em wet threshold} of $H$ as the minimum number $\tau$ of dry coordinates such that the system $[S]$
has a solution for all
$\mathbf{c}\in\mathbb{F}_2^n,\mathbf{m}\in\mathbb{F}_2^\codim$ and
\mbox{$\mathcal{W}\subseteq \{ 1,\dots,n\}$} with $\#
\mathcal{W}\le n-\tau$. The number of extra dry symbols beyond $r$, $\tau-r$ is the {\em strict overhead} of the system.
It is also of interest to compute the {\em average overhead} $\delta-r$, where $\delta$ is the average minimum number  of dry coordinates such that  $[S]$ has a solution over all possible choices of  
$H,\mathbf{c},\mathbf{m}$ and $\mathcal{W}$.
In detail, we want to determine
\begin{enumerate}
\item  necessary and sufficient conditions to ensure that the system $[S]$ has a solution; 
\item  the probability that $[S]$ has a solution for given $n,\codim$ and $\dry$; 
\item  the average overhead $\delta-\codim$ to have a solution.
\end{enumerate}
These problems have already be treated by several authors. Fridrich,
Goljan and Soukal \cite{fri1,fri2,fri3} studied (2) and showed that we can take
$\codim=\dry+O(2^{-\dry/4})$ as $\dry\rightarrow\infty$, which gives a first answer to (3).
Sch\"onfeld and  Winkler \cite{schonfeld} treated the particular
case of BCH codes, giving detailed computer results.  Barbier, Augot
and Fontaine \cite{BAF} gave sufficient conditions for the existence
of solutions by slightly modifying the problem $[S]$ for linear
codes. The case of Reed-Solomon codes has been treated by Fontaine and
Galand in \cite{fontaine+galand}.  In some of these works the reader
may find a study of the embedding efficiency as well as some
implementation issues. 

The aim of this article is to take another step in this study. We give exact answers to the  questions (1)
in section~\ref{solutions} and (3) in section~\ref{overhead} above,
relating the wet  threshold and overhead to well known parameters of the code having $H$ as parity check matrix, and 
highlighting the role played by the dual distance. The relation with
the weight hierarchy of codes is studied in section
\ref{hierarchy}. Finally in section \ref{systematic} we extend the
matrix encoding method to the broad family of
systematic codes, showing the relationship between stegoschemes,
resilient functions and orthogonal arrays. 
We show that wet paper codes arising from systematic nonlinear codes
may behave better than the ones coming from linear codes, in the sense
that they may require less free positions to ensure the existence of
solution.

\section{A necessary and sufficient condition for the existence of
  solutions}\label{solutions}

Let $\mathbf{c},\mathbf{m},H$ and $\mathcal{W}$ as defined in the
Introduction and let us study the solvability of the linear equation
\[
[S]: \left\{ \begin{array}{rl}
\mathbf{x}H^{T}&=\mathbf{m},\\
x_i& = c_i \; \mbox{ if $i\in\mathcal{W}$.}
\end{array} \right. 
\]
Let $\mathcal{C}$ be the $[n,n-\codim]$ linear code whose parity check
matrix is $H$ and let  $G$ be a generator matrix of
$\mathcal{C}$. Denote by  $\mathcal{C}^{\perp}$ the dual of
$\mathcal{C}$ and by $d^{\perp}=d(\mathcal{C}^{\perp})$ the minimum
distance of $\mathcal{C}^{\perp}$.
For $\mathbf{m}\in\mathbb{F}_2^\codim$, we
shall denote by $\mbox{cl}(\mathbf{m})$ a leader of the coset
$\{\mathbf{x}\in \mathbb{F}_2^n : \mathbf{x}H^T=\mathbf{m}\}$. 
Since $\mathbf{x}H^{T}$ can be interpreted as a syndrome, the system
$[S]$ has a solution if there exists $\mathbf{x}\in
\mbox{cl}(\mathbf{m})+\mathcal{C}$ such that
$\pi_\mathcal{W}(\mathbf{x})=\pi_\mathcal{W}(\mathbf{c})$, where
$\pi_\mathcal{W}$ is the projection over the coordinates of
$\mathcal{W}$. Equivalently $[S]$ has a solution if and only if
$\pi_{\mathcal{W}}(\mathbf{c})\in
\pi_{\mathcal{W}}(\mbox{cl}(\mathbf{m})+\mathcal{C})$.  
For a matrix $M$ with $n$ columns, let $M_\mathcal{W}$ be the matrix
obtained from $M$ by deleting the columns with indexes in $\mathcal{D}$.

\begin{lemma}
For all cosets $\mathbf{y}+\mathcal{C}$, the projections $\pi_\mathcal{W}(\mathbf{y}+\mathcal{C})$ have the same cardinality, $\# \pi_\mathcal{W}(\mathbf{y}+\mathcal{C})=2^{\mbox{\rm\small rank}(G_{\mathcal{W}})}$. Thus $[S]$ has a solution for general $\mathbf{c}$ and $\mathbf{m}$ if and only if the matrix $G_{\mathcal{W}}$ has full rank, 
$\mbox{\rm rank}(G_{\mathcal{W}})=\ell$.
\end{lemma}
\begin{proof}
$\pi_\mathcal{W}(\mathbf{y}+\mathcal{C})=\pi_\mathcal{W}(\mathbf{y})+\pi_\mathcal{W}(\mathcal{C})$, hence
$\# \pi_\mathcal{W}(\mathbf{y}+\mathcal{C})= \# \pi_\mathcal{W}(\mathcal{C})$ and  $\pi_\mathcal{W}(\mathcal{C})$ is a vector space of dimension  $\mbox{\rm rank}(G_{\mathcal{W}})$. For fixed $\mathbf{c}$ and $\mathbf{m}$,  $[S]$ has a solution if and only if $\pi_{\mathcal{W}}(\mathbf{c})\in \pi_{\mathcal{W}}(\mbox{cl}(\mathbf{m})+\mathcal{C})$. Since $\#\pi_{\mathcal{W}}(\mathbb{F}_2^n)=2^\ell$, this occurs for all  $\mathbf{c}$ and $\mathbf{m}$ if and only if  $G_{\mathcal{W}}$ has full rank, $\mbox{\rm rank}(G_{\mathcal{W}})=\ell$. Note that $\codim\le \dry$.
\end{proof}

\begin{lemma}
$G_{\mathcal{W}}$ has full rank if and only if there is no nonzero word of $\mathcal{C}^{\perp}$ with support contained in $\mathcal{W}$.
\end{lemma}
\begin{proof}
Since $G$ is a parity check matrix of $\mathcal{C}^{\perp}$, a nonzero word in $\mathcal{C}^{\perp}$ with support contained in $\mathcal{W}$ imposes a linear condition on the columns of $G_{\mathcal{W}}$ and conversely.
\end{proof}

More generally, if there exist  $w$ independent words of $\mathcal{C}^{\perp}$ with support in $\mathcal{W}$ then   we have $\mbox{rank}(G_{\mathcal{W}})=\ell-w$. This suggests that the weight hierarchy of  $\mathcal{C}$ also plays a role in the study of the solvability of $[S]$.   This study will be conducted later in section \ref{hierarchy}.

\begin{theorem}\label{solutioniff}
The system $[S]$ has a solution for arbitrary
$\mathbf{c}\in\mathbb{F}_2^n,\mathbf{m}\in\mathbb{F}_2^\codim$ and
\mbox{$\mathcal{W}\subseteq \{ 1,\dots,n\}$} with $\#
\mathcal{W}=n-\dry$, if and only if $\dry\ge n-d^\perp+1$. In this
case $[S]$ has exactly $2^{\dry-\codim}$ solutions.
\end{theorem}
\begin{proof}
If $\dry\ge n-d^{\perp}+1$ then $\# \mathcal{W}<d^{\perp}$ and no
nonzero codeword of $\mathcal{C}^\perp$ has support contained in
$\mathcal{W}$. Conversely, take a codeword of weight $d^{\perp}$ and a
set $\mathcal{W}$ of cardinality $n-\dry$ containing its support. Then
$\mbox{rank}(G_{\mathcal{W}})<n-\dry$ and the homogeneous system
$H\mathbf{x}^t=\mathbf{0}$ has no solution for  $\mathbf{c}$ such that
$\pi_{\mathcal{W}}(\mathbf{c})$ is not in the subspace
$\pi_{\mathcal{W}}(\mathcal{C})$ spanned by the rows of
$G_{\mathcal{W}}$.
When $\mbox{rank}(G_{\mathcal{W}})=n-\dry$, then the number of
solutions is $\# \mathcal{C} /\#
\pi_{\mathcal{W}}(\mathcal{C})=2^{\dry-\codim}$.
\end{proof}

Then when using a parity check matrix of a $[n,n-\codim]$ code
$\mathcal{C}$,  at most $n-d^{\perp}+1$ dry symbols are needed to
embed $\codim$ information symbols. The wet threshold of $\mathcal{C}$ is 
$\tau=n-d^{\perp}+1$ and its strict overhead is $n-d^{\perp}+1- \codim$.
Remark that according to the Singleton
bound applied to $\mathcal{C}^{\perp}$ we have $n-d^{\perp}+1\ge
\codim$. The difference $n-d^{\perp}+1- \codim$ is known as the {\em Singleton
defect} of $\mathcal{C}^{\perp}$. Thus, when using a parity check
matrix of  $\mathcal{C}$ to embed information via wet paper codes, the
strict overhead is just the
Singleton defect of the dual code $\mathcal{C}^{\perp}$.

\begin{example} \label{ag}
{\rm {\bf (1)} Consider the binary Hamming code of redundancy $s$ and length
  \mbox{$n=2^s-1$}. The dual distance is $d^{\perp}=2^{s-1}$, hence we can
  embed $s$ information bits into a cover vector of length $n$ with
  $2^{s-1}\approx n/2$ dry positions. To
  see that less dry symbols are
  not enough to have solution with certainty, consider a parity-check
  matrix whose rows are the binary representations of integers
  $1,\dots,2^{s}-1$. When deleting the last $2^{s-1}$ columns of $H$
  we obtain a matrix whose last row is 0. Note also that the method proposed in
  \cite{BAF} allows to embed one information bit for $n/2$ dry positions
  modifying one bit of the cover vector.
  {\bf (2)} In general it is not simple to construct codes with bounded Singleton defect. An exception are algebraic geometry codes, built from an algebraic curve and two rational divisors, see \cite{abundance}. It is known that the Singleton defect of a code comming from a curve $\mathcal{X}$ is bounded by the genus of $\mathcal{X}$. Therefore it is possible to construct wet paper codes with strict overhead as small as desired.}
\end{example}

\section{Computing the overhead}\label{overhead}

Our second task is to compute the average overhead
$\tilde{m}=\dry-\codim$ to have a solution for random
$\mathcal{C}, \mathcal{W}, \mathbf{c}$ and $\mathbf{m}$ (according
with previous notations). Also we obtain an estimate on
the probability of having  solution.  Let us denote by
$\mbox{avrank}(t,s)$ the average rank of a random $t\times s$ matrix
$M$. 

\begin{proposition}
For random $\mathcal{C}, \mathcal{W}, \mathbf{c}$ and $\mathbf{m}$ as above, the probability that $\delta$ dry symbols are enough to transmit $\codim\le \delta$ message symbols is
\[
p=2^{\mbox{\rm\small avrank}(n-\codim,n-\dry)-(n-\dry)}.
\]
\end{proposition}
\begin{proof}
The probability that the corresponding system $[S]$ have a solution is
\[
p=\mbox{prob}\left( \pi_{\mathcal{W}}(\mathbf{c})\in \pi_{\mathcal{W}}(\mbox{cl}(\mathbf{m})+\mathcal{C})\right)
=\frac{\# \pi_{\mathcal{W}}(\mathcal{C})}{2^{n-\dry}} = \frac {2^{\mbox{\rm\small avrank}( G_{\mathcal{W}}) }}{2^{n-\dry}}.
\]
\end{proof}

The function $\mbox{avrank}(t,s)$ can be computed using theorem \ref{rangos}  below.
The rank properties of random matrices have been investigated in coding theory, among other fields, related to codes for the erasure channel, see e.g. \cite{rank3}. As shown in \cite{fri1,fri2}, these results allow us to give an estimate on the average overhead. Since $G_{\mathcal{W}}$ is a $(n-\codim)\times (n-\dry)$ matrix and $(n-\codim)-(n-\dry)=\dry-\codim$, then $\tilde{m}$ can be seen as the minimum number of extra rows beyond $n-\dry$ required to obtain a matrix of full rank.
Let $t,m$ be non negative integers and $M_{t+m,t}$ be a random $(t+m)\times t$ matrix with $m\ge 0$.

\begin{theorem}\label{rangos}
Let $M_{t+m,t}$ be a matrix where the elements of $\mathbb{F}_2$ are equally likely. Then 
\[
\lim_{t\rightarrow\infty}{\rm prob}\left( {\rm rank}(M_{t+m,t})=t-s\right) = \prod_{j=s+m+1}^{\infty}\left( 1-2^{-j}\right)/2^{s(s+m)}\prod_{j=1}^{s}\left( 1-2^{-j}\right).  
\]
\end{theorem}

See \cite{rank1,rank2,rank3}.
It is known that this formula is very accurate even for small $t$ and $m$.
This theorem directly allows us to obtain numerical
estimates on the function $\mbox{avrank}$ and consequently on the
probability that $[S]$ admits a solution for random $\mathcal{C},
\mathcal{W}, \mathbf{c}$ and $\mathbf{m}$. These estimates can be
found in the literature (see \cite{schonfeld} and the references
therein) and we will not repeat them here.
Also theorem \ref{rangos} can be used to  compute
the average number of extra rows needed to have full rank. Following
\cite{rank3}, for any positive $m$, let 
\[
Q_m=\prod_{j=m+1}^{\infty}\left(1- 2^{-j}\right).
\] 
According to the  theorem, the probability that exactly $m$ extra rows beyond $t$ are needed to
obtain a $(t+m)\times t$ random matrix of full rank is
$Q_m-Q_{m-1}$. Since \mbox{$Q_{m-1}=((2^m-1)/2^m)Q_m$} we have
$Q_m-Q_{m-1}=Q_m/2^m$, so the average number of extra rows is 
\[
\tilde{m}=\sum_{m=1}^{\infty} m(Q_m-Q_{m-1})=\sum_{m=1}^{\infty}\frac{m}{2^m}Q_m.
\]
This series is convergent as it is upper-bounded by a convergent
arithmetic-geometric series.
Let us remember that from elementary calculus we have $\sum_{m=1}^{\infty} m x^m =x/(1-x)^2$ when $|x|<1$. Then
\[
\tilde{m}=\sum_{m=1}^{\infty}\frac{m}{2^m}Q_m < \sum_{m=1}^{\infty}\frac{m}{2^m}=2.
\]
A direct computation shows that  $\tilde{m}= 1.6067..$. Then  the average overhead is $1.6067$ and, for $n$ large enough, $\dry$ dry bits are enough to transmit $\codim\approx \dry-1.6$ information bits.

\section{Solvability and the generalized Hamming weights} \label{hierarchy}

Let $\mathcal{C}$ be a linear $[n,n-\codim]$ code and let
$\mathcal{C}^{\perp}$ be its dual. The dual distance $d^\perp$ can be
expressed in terms of $\mathcal{C}$  via its weight hierarchy. Let us
remember that for $1\le \MDSr \le n-\codim$, the $\MDSr$-th generalized
Hamming weight of $\mathcal{C}$ is defined as (see \cite{wei})
$$
d_\MDSr(\mathcal{C})= \min \{ \# \mbox{supp}(L) :  \mbox{$L$ is a
  $\MDSr$-dimensional linear subspace of $\mathcal{C}$}\}
$$
where $\mbox{supp}(L)=\cup _{\mathbf{x}\in L}
\mbox{supp}(\mathbf{x})$. The sequence $d_1(\mathcal{C}),
\dots,d_{n-\codim}(\mathcal{C})$ is the {\em weight hierarchy} of
$\mathcal{C}$.  Two important properties of the weight hierarchy  are
the monotonicity
\mbox{$d_1(\mathcal{C})<d_2(\mathcal{C})<\dots<d_{n-\codim}(\mathcal{C})$}
and the duality
\mbox{$
\{ d_1(\mathcal{C}),\dots,d_{n-\codim}(\mathcal{C})\} \cup \{n+1-
d_1(\mathcal{C}^\perp), \dots,
n+1-d_\codim(\mathcal{C}^\perp)=\{1,\dots,n\}
$}. For simplicity we shall write $d_1,\dots,d_{n-\codim}$ and
$d_1^\perp,\dots,d_\codim^\perp$.
If $d_{n-\codim}=n$, we define the {\em MDS rank} of $\mathcal{C}$ as the least integer $\MDSr$ such that
$d_\MDSr =\codim+\MDSr$ (and consequently $d_s=\codim+s$ for all $s\ge
\MDSr$). Note that classical MDS codes are first rank MDS codes.

\begin{proposition}
If $\mathcal{C}$ has MDS rank $\MDSr$ and $\dry\ge \MDSr +\codim-1$,
then the corresponding system $[S]$ has a solution for arbitrary
$\mathbf{c}\in\mathbb{F}_2^n,\mathbf{m}\in\mathbb{F}_2^\codim$ and
$\mathcal{W}\subseteq \{ 1,\dots,n\}$ with $\# \mathcal{W}=n-\dry$.
\end{proposition}
\begin{proof}
By the duality property $\mathcal{C}$ has MDS rank
$\MDSr=n-\codim-d^\perp+2$ hence $n-d^\perp+1=\MDSr+\codim-1$ and
proposition~\ref{solutioniff} implies the result.
\end{proof}

This proposition leads us to consider codes with low MDS rank. MDS
codes, and Reed-Solomon codes in particular, were proposed as good
candidates in \cite{fontaine+galand}. The main drawback of MDS codes
is its small length. So, we may consider codes of higher rank,
reaching a balance between length and security in the existence of
solutions. In this sense, algebraic geometry codes (defined in example \ref{ag}) can be a good
option. It is known that an AG code coming from a curve of genus $g$
has MDS rank at most $g+1-a$, where $a$ is its abundance, more details
in \cite{abundance}. Yet we find again the problem of the length of
obtained codes. For example, it has been conjectured that Near MDS
codes (codes  for which $d+d^\perp=n$) over $\mathbb{F}_q$ have length
upper bounded by  $q+1+2\sqrt{q}$ (observe that codes arising from
elliptic curves are either MDS or NMDS). Another option is to extend
the ground alphabet. Several strategies have been proposed. One of the
more interesting is to consider the Justensen construction with
algebraic geometric codes \cite{Shen}.  The following result
extends proposition~\ref{solutioniff} to all generalized Hamming
weights.

\begin{proposition}
If $d_\MDSr>\dry\ge \codim$ for some $\MDSr\ge \dry-\codim$, then
$\mbox{rank}(G_{\mathcal{W}})\ge n-\codim-\MDSr+1$ for every set
$\mathcal{W}\subseteq \{ 1,\dots,n\}$ with $\# \mathcal{W}=n-\dry$.
\end{proposition}
\begin{proof}
Consider the code $\mathcal{C}^\perp_\mathcal{W}$  obtained from $\mathcal{C}^\perp$ by shortening at the positions in $\mathcal{D}$. Since $G_\mathcal{W}$ is a parity-check matrix for $\mathcal{C}^\perp_\mathcal{W}$, we have
$
\mbox{rank}(G_\mathcal{W})=n-\dry-\dim (\mathcal{C}^\perp_\mathcal{W})
$.
If $d_\MDSr^\perp\ge n-\dry+1$ then it holds that
$\dim(\mathcal{C}^\perp_\mathcal{W})\le \MDSr-1$, hence
$\mbox{rank}(G_\mathcal{W})\ge n-\dry-\MDSr+1$. Assume
$d_\MDSr>\dry$. Then $n-d_\MDSr+1<n-\dry+1$, and by the duality and
monotonicity properties, the interval $[n-\dry+1,n]$ contains at least
$\dry-\MDSr+1$ terms of the weight hierarchy of
$\mathcal{C}^\perp$. Thus $d^\perp_{\codim-\dry+\MDSr}\ge n-\dry+1$
and we get the statement.
\end{proof}

\section{A generalization to systematic codes}\label{systematic}

In this section we extend the matrix embedding construction, and the
wet paper method in particular, to the wide family of systematic
codes. We show that stegoschemes based on these codes are handled
essentially in the same manner as in the case of linear codes. The use
of systematic codes was suggested by Brierbauer and Fridrich in
\cite{covering}, where stegoschemes arising from the
Nordstrom-Robinson codes are treated in some detail. Here we go deeper
into this study, showing the relationships between stegoschemes,
orthogonal arrays and resilient functions.
We pay special attention to the analogue of proposition
\ref{solutioniff}, showing its combinatorial nature.

\subsection{Systematic codes}\label{SystematicCodes}

Let us remember that for a set $\mathcal{U}\subseteq \{ 1,\dots,n\}$
with $u=\#\mathcal{U}$, we denote by
$\pi_{\mathcal{U}}:{\mathbb{F}_2}^n\rightarrow{\mathbb{F}_2}^u$ the
projection on the coordinates of $\mathcal{U}$. If $\mathcal{V}=
\{1,\dots,n\} \setminus\mathcal{U}$, we shall write a vector
$\mathbf{x}\in\mathbb{F}_2^n$ as ${\mathbf x}=({\mathbf u},{\mathbf
  v})$, where ${\mathbf u}=\pi_{\mathcal{U}}({\mathbf
  x})\in\mathbb{F}_2^u$ and ${\mathbf v}=\pi_{\mathcal{V}}({\mathbf
  x})\in \mathbb{F}_2^v$, $v=n-u$.
A code $\mathcal{C}$ of length $n$ is {\em systematic} if there exist
$u$ positions that carry the information. More formally, given a set
$\mathcal{U}\subseteq \{ 1,\dots,n\}$, we say that $\mathcal{C}$ is
{\em systematic at the positions of $\mathcal{U}$} (or simply {\em
  systematic} when the set $\mathcal{U}$ is understood) if for every
$\mathbf{u}\in\mathbb{F}_2^u$ there exists one and only one codeword
$\mathbf{x}\in\mathcal{C}$ such that
$\pi_{\mathcal{U}}(\mathbf{x})=\mathbf{u}$. Up to reordering of
coordinates we can always assume that $\mathcal{U}=\{1,\dots,u\}$ and
$\mathcal{V} = \lbrace u+1,\dots,n \rbrace$.

If $\mathcal{C}$ is systematic then $\# \mathcal{C}=\#
\mathbb{F}_2^{u}=2^u$. We say that $\mathcal{C}$ is a $[n,u]$
code. Clearly every $[n,u]$ linear code is systematic of dimension $u$
hence this notation is consistent. Thus systematic codes generalize
linear codes. However the family of systematic codes is much greater
than the family of linear codes (apart from the advantage of being
defined over alphabets other than fields). To see that note that it is
fairly simple to construct a  systematic code $\mathcal{C}$: just
complete each vector in $\mathbb{F}_2^u$ to a vector in
$\mathbb{F}_2^n$. This completion induces a {\em generator function}
$\sigma=\sigma_{\mathcal{C}}:\mathbb{F}_2^u\rightarrow\mathbb{F}_2^{n-u}$
defined to  $\mathcal{C}=\{ (\mathbf{u},\sigma(\mathbf{u})) :
\mathbf{u}\in\mathbb{F}_2^u \}$. Then $\mathcal{C}$ is linear if and
only if so is $\sigma$. In this case there exists a $u\times(n-u)$
matrix $\Sigma$ such that $\sigma(\mathbf{u})=\mathbf{u}\Sigma$. Then
$(I_u,\Sigma)$ is a generator matrix for $\mathcal{C}$ and
consequently $H=(-\Sigma^T,I_{n-u})$ is a parity-check matrix of
$\mathcal{C}$. Since every map $\mathbb{F}_2^u\rightarrow
\mathbb{F}_2$ can be written as a reduced polynomial, the components
$\sigma_1,\dots,\sigma_{n-u}$ of $\sigma$ are square free reduced polynomials in
variables $x_1,\dots,x_u$.

The family of systematic codes contains some nonlinear codes having excellent parameters. Among these we can highlight the Preparata, Kerdrock, Nodstrom-Robinson and many others. Some of them  have also efficient decoding systems (which is the main drawback of nonlinear codes). Other well known example is the following.

\begin{example}\label{Nadler}
The Nadler code $\mathcal{N}$ is a $[12,5]$ systematic nonlinear code with covering radius $\rho=4$ and minimum distance $d=5$, \cite{Nad2}.  $\mathcal{N}$ contains twice as many codewords as any linear code with the same length and minimum distance, see \cite{vLint}. Among the current practical applications of $\mathcal{N}$ we can mention its use for the decoder module of  SINCGARS radio systems \cite{Sinc1}.
The combinatorial structure of $\mathcal{N}$ was shown by van Lint in \cite{NadLint}; following this article, the 32 codewords of $\mathcal{N}$ are shown in Table 1. 

\begin{table}[htbp]\label{TNadler}
\begin{center}
\begin{tabular}{ccc}
\begin{tabular}{|cccc|}\hline
011&100&100&100 \\
101&010&010&010\\
110&001&001&001\\ \hline
100&011&100&100\\
010&101&010&010\\
001&110&001&001\\ \hline
100&100&011&100\\
010&010&101&010\\
001&001&110&001\\ \hline
100&100&100&011\\
010&010&010&101\\
001&001&001&110\\ \hline
\end{tabular} &
\begin{tabular}{|cccc|}\hline
111&010&100&001\\
111&001&010&100\\
111&100&001&010\\\hline
010&111&001&100\\
001&111&100&010\\
100&111&010&001\\ \hline
100&001&111&010\\
010&100&111&001\\
001&010&111&100\\ \hline
001&100&010&111\\
100&010&001&111\\
010&001&100&111\\ \hline
\end{tabular} &
\begin{tabular}{|cccc|}\hline
011&011&011&011\\
101&101&101&101\\
110&110&110&110\\ \hline
000&111&111&111\\
111&000&111&111\\
111&111&000&111\\
111&111&111&000\\ \hline
000&000&000&000\\ \hline
\end{tabular}
\end{tabular}
\newline\vspace*{0.5cm}
Table 1. The 32 codewords of the Nadler code
\end{center}
\end{table}

It is not hard to check that this code is systematic at positions 1,2,4,7,  10. Besides the exhaustive enumeration given in Table 1, $\mathcal{N}$ can be described by the function $\sigma$. Up to reordering of coordinates so that $\mathcal{N}$ is systematic at positions $1,\dots,5$, we have

\begin{eqnarray*}
\sigma_6    &=& x_1+x_2+x_3+(x_1+x_5)(x_3+x_4) \\
\sigma_7    &=& x_1+x_2+x_4+(x_1+x_3)(x_4+x_5) \\
\sigma_8    &=& x_1+x_2+x_5+(x_1+x_4)(x_3+x_5) \\
\sigma_9    &=& x_2+x_3+x_4+x_1x_4+x_4x_5+x_5x_1 \\
\sigma_{10} &=& x_2+x_3+x_5+x_1x_3+x_3x_4+x_4x_1 \\
\sigma_{11} &=& x_1+x_4+x_5+x_1x_3+x_3x_5+x_5x_1 \\
\sigma_{12} &=& x_1+x_2+x_3+x_4+x_5+x_3x_4+x_4x_5+x_5x_3.
\end{eqnarray*}
\end{example}

In order to construct a stegoscheme from a systematic code, we need a partition of $\mathbb{F}_2^n$ and a family of decoding maps, one map for each element of the partition.

\begin{proposition}\label{particion}
Let $\mathcal{C}$ be a $[n,u]$ systematic code. Then the sets $(\mathbf{0},\mathbf{v})+\mathcal{C}$,  $\mathbf{v}\in\mathbb{F}_2^{n-u}$, are pairwise disjoint and hence the family $((\mathbf{0},\mathbf{v})+\mathcal{C} : \mathbf{v}\in\mathbb{F}_2^{n-u})$ is a partition of $\mathbb{F}_2^n$.
\end{proposition}
\begin{proof}
If $(\mathbf{0},\mathbf{v}_1)+\mathbf{c}_1=(\mathbf{0},\mathbf{v}_2)+\mathbf{c}_2$ for some $\mathbf{v}_1,\mathbf{v}_2\in\mathbb{F}_2^{n-u}$ and $\mathbf{c}_1,\mathbf{c}_2\in\mathbb{F}_2^n$, then \mbox{$\pi_{\mathcal{U}}(\mathbf{c}_1)= \pi_{\mathcal{U}}((\mathbf{0},\mathbf{v}_1)+\mathbf{c}_1)=\pi_{\mathcal{U}}((\mathbf{0},\mathbf{v}_2)+\mathbf{c}_2)=\pi_{\mathcal{U}}(\mathbf{c}_2)$}. Then $\mathbf{c}_1=\mathbf{c}_2$ and consequently $\mathbf{v}_1=\mathbf{v}_2$.
\end{proof}

In general the translates $\mathbf{x}+\mathcal{C}$ are not pairwise
disjoint when $\mathbf{x}$ runs over the whole space
$\mathbb{F}_2^n$. In that case these sets do not give a partition of
$\mathbb{F}_2^n$ and they are not useful for decoding purposes. Anyway
the partition given by proposition \ref{particion} allows us to define
a {\em syndrome map} $s:\mathbb{F}_2^n\rightarrow \mathbb{F}_2^{n-u}$
as follows: define $s(\mathbf{x})=\mathbf{v}$ if $\mathbf{x}\in
(\mathbf{0},\mathbf{v})+\mathcal{C}$. The systematic property leads us
to compute $s(\mathbf{x})$ efficiently: if
$\mathbf{x}=(\mathbf{u},\mathbf{w})$ then we can write
$(\mathbf{u},\mathbf{w})=(\mathbf{u},\sigma(\mathbf{u}))+(\mathbf{0},s(\mathbf{x}))$
and hence $s(\mathbf{x})=\mathbf{w}-\sigma(\mathbf{u})$. If $\sigma$
is a linear map, and hence the code $\mathcal{C}$ is linear, then
$s(\mathbf{x})=\mathbf{x} H^T$ is the usual syndrome for linear codes.

A {\em decoding map} for a general code $\mathcal{C}\subseteq\mathbb{F}_2^n$ is a mapping $\mbox{dec}:\mathbb{F}_2^n\rightarrow\mathcal{C}$ such that for every $\mathbf{x}\in\mathbb{F}_2^n$, $\mbox{dec}(\mathbf{x})$ is the closest word to $\mathbf{x}$ in $\mathcal{C}$. If more than one of such words exists, simply choose one of them at random. 

\begin{proposition}\label{decoclases}
Let $\mathcal{C}\subseteq\mathbb{F}_2^n$ be a systematic code and $\mathbf{z}\in\mathbb{F}_2^n$. If $\mbox{\rm dec}$ is a decoding map for $\mathcal{C}$ then $\mbox{\rm dec}_\mathbf{z}(\mathbf{x})=\mathbf{z}+\mbox{\rm dec}(\mathbf{x}-\mathbf{z})$ is a decoding map for the code $\mathbf{z}+\mathcal{C}$.
\end{proposition}
\begin{proof}
Clearly $\mbox{dec}_\mathbf{z}(\mathbf{x})\in \mathbf{z}+\mathcal{C}$. If there exists $\mathbf{z}+\mathbf{c}\in\mathbf{z}+\mathcal{C}$ such that \mbox{$d(\mathbf{x},\mathbf{z}+\mathbf{c})< d(\mathbf{x},\mbox{dec}_\mathbf{z}(\mathbf{x}))$}, then $d(\mathbf{x}-\mathbf{z},\mathbf{c})< d(\mathbf{x}-\mathbf{z},\mbox{dec}(\mathbf{x}-\mathbf{z}))$, which contradicts that dec is a decoding map for $\mathcal{C}$.
\end{proof}

\subsection{Stegoschemes from systematic codes}

Let $\mathcal{C}$ be a $[n,n-u]$ systematic code and let ${\rm dec}$ be a decoding map for $\mathcal{C}$. 
According to propositions \ref{particion} and \ref{decoclases}, we obtain a stegoscheme $\mathcal{S}=\mathcal{S}(\mathcal{C})$ from $\mathcal{C}$, whose embedding and recovering maps are
$$
{\rm emb}:\mathbb{F}_2^n\times\mathbb{F}_2^{u}\rightarrow\mathbb{F}_2^n \; , \; 
{\rm emb}(\mathbf{c},\mathbf{m}) = {\rm dec}_{(\mathbf{0},\mathbf{m})}(\mathbf{c})= 
(\mathbf{0},\mathbf{m})+{\rm dec}(\mathbf{c}-(\mathbf{0},\mathbf{m}))
$$  
$$
{\rm rec}:\mathbb{F}_2^n\rightarrow \mathbb{F}_2^{u} \; , \; 
{\rm rec}(\mathbf{x})=s(\mathbf{x}).
$$
By definition of syndrome it holds that ${\rm rec}({\rm emb}(\mathbf{c},\mathbf{m})) = s({\rm dec}_{(\mathbf{0},\mathbf{m})}(\mathbf{c}))=\mathbf{m}$ for all $\mathbf{c}\in\mathbb{F}_2^n$ and $\mathbf{m}\in\mathbb{F}_2^{u}$. Compare this with the usual expression ${\rm emb}(\mathbf{c},\mathbf{m})= \mathbf{c}-{\rm cl}(\mathbf{c}H^T-\mathbf{m})$  for  linear codes.  We note that to perform this embedding it is necessary to have a table with all syndromes and cosets leaders, even if the decoding map used does not require them. Therefore, the systematic formulation can be useful even using linear codes. 

Let us study the parameters of $\mathcal{S}(\mathcal{C})$ in relation with those of $\mathcal{C}$. The cover length is $n$ and the embedding capacity $\codim=u$.
To compute its embedding radius and average number of embedding changes we first need to recall some concepts from coding theory.  Given a general code $\mathcal{D}$, its {\em covering radius} is defined as the maximum distance from a vector $\mathbf{x}\in\mathbb{F}_2^n$ to $\mathcal{D}$,
$\rho(\mathcal{D})=\max\{ d(\mathbf{x},\mathcal{D}) : \mathbf{x}\in\mathbb{F}_2^n  \}$,
where $d(\mathbf{x},\mathcal{D})=\min\{ d(\mathbf{x},\mathbf{c}) : \mathbf{c}\in\mathcal{D} \}$. The {\em average radius} of $\mathcal{D}$, $\tilde{\rho}(\mathcal{D})$ is the average distance from a vector $\mathbf{x}\in\mathbb{F}_2^n$ to $\mathcal{D}$
$$
\tilde{\rho}(\mathcal{D})=\frac{1}{2^n}\sum_{\mathbf{x}\in\mathbb{F}_2^n} d(\mathbf{x},\mathcal{D}). 
$$ 
If $\mathcal{D}$ is linear then both parameters can be obtained from
the {\em coset leader distribution} of $\mathcal{D}$, that is the
sequence $\alpha_0,\dots,\alpha_n$, where $\alpha_i$ is the number of
coset leaders of weight $i$. Clearly $\alpha_i\le
\binom{n}{i}$. When $i\le t=\lfloor
(d(\mathcal{D})-1)/2\rfloor$ then all vectors of weight $i$ are
leaders hence we get equality, $\alpha_i=\binom{n}{i}$. For
$i>t$ the computation of $\alpha_i$ is a classical problem, considered
difficult. For nonlinear $\mathcal{D}$ the coset leader distribution
may be generalized to the distribution of distances to the code,
defined as
$$
\alpha_i=\frac{1}{\# \mathcal{D}} \# \{\mathbf{x}\in\mathbb{F}_2^n :
d(\mathbf{x}, \mathcal{D})=i  \}
$$ 
If $\mathcal{D}$ is linear then both definitions of $\alpha_i$'s coincide.
A similar reasoning as above shows that the property $\alpha_i\le
\binom{n}{i}$ with equality when $i\le t=\lfloor
(d(\mathcal{D})-1)/2\rfloor$ remains true for all codes.  The covering
radius of $\mathcal{D}$ is the maximum $i$ such that $\alpha_i\neq 0$
and the average radius is given by
$$
\tilde{\rho}(\mathcal{D})=\frac{\# \mathcal{D}}{2^n}\sum_{i=0}^{n}
i\alpha_i.
$$
If  $\mathcal{C}$ is $[n,n-u]$ systematic, then  for all $\mathbf{v}\in\mathbb{F}_2^{u}$ and $\mathbf{x}\in\mathbb{F}_2^n$ we have
\mbox{$d(\mathbf{x},(\mathbf{0},\mathbf{v})+\mathcal{D})=d(\mathbf{x}-(\mathbf{0},\mathbf{v}),\mathcal{D})$}. Thus 
all the
translates \mbox{$((\mathbf{0},\mathbf{v})+\mathcal{D} :
\mathbf{v}\in\mathbb{F}_2^{u})$} have the same  distribution of
distances to the code and hence the same average radius and covering
radius. As a consequence we have the following result, which is well known for stegoschemes comming from linear codes.

\begin{proposition}
Let $\mathcal{C}$ be a $[n,n-u]$ systematic code and let $\mathcal{S}$ the stegoscheme obtained from $\mathcal{C}$. Then
the embedding radius of $\mathcal{S}$ is the covering radius of $\mathcal{C}$ and
the average number of embedding changes $R_a(\mathcal{S})$ is the average radius of $\mathcal{C}$.
\end{proposition}
\begin{proof}
By definition of decoding map, the number of changes when embedding a message $\mathbf{m}$ into a vector $\mathbf{x}$ is $d(\mathbf{x},{\rm emb(\mathbf{x},\mathbf{m})})=d(\mathbf{x}-(\mathbf{0},\mathbf{m}),\mathcal{C})$ and both statements hold.
\end{proof}

\begin{example}
The Nadler code of Example \ref{Nadler} is 2-error correcting, hence
$\alpha_0=1$, $\alpha_1=12$, $\alpha_2=66$. Other values of $\alpha$ are obtained by computer search: $\alpha_3=46$, $\alpha_4=3$, and $\alpha_i=0$ for $i=5,\dots,12$. Then $\tilde{\rho}(\mathcal{N})=2.296875$. The stegoscheme derived from $\mathcal{N}$ allows to embed 7 bits of information into a cover vector of 12 bits, by changing $2.296875$ of them on average and 4 of them at most. 
\end{example}

\subsection{Stegoschemes, resilient functions and orthogonal arrays}\label{ResilitientOA}

Systematic codes allows us to make a connection of stegoschemes
with two objects of known importance in information theory: resilient
functions and orthogonal arrays.  A function $f:\mathbb{F}_2^n
\rightarrow\mathbb{F}_2^\codim$ is called {\em t-resilient} for some
integer $t\le n$, if for every $\mathcal{T}\subseteq \{1,\dots,n\}$
such that $\# \mathcal{T}=t$ and every $\mathbf{t}\in \mathbb{F}_2^t$,
all possible outputs of $f(\mathbf{x})$ with
$\pi_{\mathcal{T}}(\mathbf{x})=\mathbf{t}$ are equally likely to
occur, that is if for all $\mathbf{y}\in\mathbb{F}_2^{\codim}$ we have
$$
\mbox{prob} (f(\mathbf{x})=\mathbf{y} \; | \; \pi_{\mathcal{T}}(\mathbf{x})=\mathbf{t})=\frac{1}{2^\codim}
$$ 
(see the relation to recovering maps of stegoschemes).
Resilient functions play an important role in cryptography, and are closely related to orthogonal arrays  \cite{stinson1}. An {\em orthogonal array} $OA_\lambda(t,n)$ is a $\lambda 2^t\times n$ array over $\mathbb{F}_2$, such that in any $t$ columns every one of the possible $2^t$ vectors of $\mathbb{F}_2^t$ occurs in exactly $\lambda$ rows. A {\em large set} of orthogonal arrays is a set of $2^{n-t}/\lambda$ arrays  $OA_\lambda(t,n)$ such that every vector of $\mathbb{F}_2^n$ occurs once as a row of one $OA$ in the set. Then, by considering the rows of these arrays as vectors of $\mathbb{F}_2^n$, a large set of $OA$ gives a partition of $\mathbb{F}_2^n$, see \cite{stinson1}.

There is a fruitful connection between orthogonal arrays and codes, see \cite{MacSlo} Chapter 5, section 5.
If $\mathcal{C}$ is a linear $[n,n-\codim]$ code then, according to proposition \ref{solutioniff}, the array having the codewords of $\mathcal{C}$ as rows is an $OA_{2^{n-\codim-d^{\perp}+1}}(d^\perp-1,n)$. Delsarte observed that a similar result holds also for nonlinear codes. Of course if $\mathcal{C}$ is not linear then the dual code does not exist, but the dual distance can be defined  from the distance distribution of $\mathcal{C}$ via the dual transforms as follows \cite{MacSlo}: The distance distribution of $\mathcal{C}$
is defined to be the sequence $A_0,\dots,A_n$, where
$$
A_i=\frac{1}{\# \mathcal{C}} \# \{ (\mathbf{x},\mathbf{y})\in\mathcal{C}^2 : d(\mathbf{x},\mathbf{y})=i  \}
$$
$i=0,\dots,n$. The dual distance distribution of $\mathcal{C}$ is $A_0^\perp,\dots,A_n^\perp$, where 
$$
A_i^\perp=\frac{1}{\# \mathcal{C}} \sum_{j=0}^n A_jK_i(j)
$$ 
and $K_i(x)$ is the $i$-th Krautchouk polynomial 
$$
K_i(x)=\sum_{j=0}^i (-1)^j \binom{x}{j}\binom{n-x}{i-j}.
$$
If $\mathcal{C}$ is linear then $A_0^\perp,\dots,A_n^\perp$ is the distance distribution of $\mathcal{C}^\perp$. 
If $\mathcal{C}$ is $[n,n-u]$ systematic, then the array
having the codewords of $\mathcal{C}$ as rows is an
$OA_{2^{n-u-d^\perp+1}}(d^\perp-1,n)$. Furthermore in this case all
the translates $(\mathbf{0},\mathbf{v})+\mathcal{C}$ have the same
distance distribution and hence the same dual distance.

\begin{proposition}\label{OAres}
Let $\mathcal{C}$ be a systematic $[n,n-u]$ code with generator function $\sigma$ and dual distance $d^\perp$. For any $\mathbf{v}\in\mathbb{F}_2^{u}$ let $M_{\mathbf{v}}$ be the array having the words of $(\mathbf{0},\mathbf{v})+\mathcal{C}$ as rows. Then \newline
(a) The set $\{M_{\mathbf{v}} \; | \; \mathbf{v}\in\mathbb{F}_2^{u}\}$ is a large set of $OA_{2^{n-u-d^\perp+1}}(d^\perp-1,n)$. \newline
(b) The syndrome map $s(\mathbf{u},\mathbf{w})=\mathbf{w}-\sigma(\mathbf{u})$ is an $(d^\perp-1)$-resilient function. 
\end{proposition}
\begin{proof}
Let $\mathcal{T}\subseteq \{1,\dots,n\}$ with $\#\mathcal{T}=d^\perp-1$ and $\mathbf{t}\in\mathbb{F}_2^t$. (a) According to Delsarte's theorem, every one of the possible $2^t$ vectors  $\mathbf{t}$ occurs in exactly $2^{n-u-d^\perp+1}$ rows of $\pi_{\mathcal{T}}(\mathcal{C})$. Then the same happens in each of the translates $(\mathbf{0},\mathbf{v})+\mathcal{C}$. 
(b) As a consequence of (a), all possible outputs of $s(\mathbf{x})$
with $\pi_{\mathcal{T}}(\mathbf{x})=\mathbf{t}$ are equally likely to
occur, \cite{stinson1}.
\end{proof}

\subsection{Locked positions with systematic codes}

Let us return to the problem of embedding with locked positions.
Let $\mathbf{c}\in\mathbb{F}_2^n$ be a cover vector and
$\mathbf{m}\in\mathbb{F}_2^u$ be the secret we want to embed into
$\mathbf{c}$. There is a set $\mathcal{W}\subseteq\{1,\dots,n\}$ of
$n-\dry$ locked positions that cannot be altered during the embedding
process. Consider a systematic $[n,n-u]$ code $\mathcal{C}$ and let
$s$ be the syndrome of $\mathcal{C}$ defined in section
\ref{SystematicCodes}. As in the case of linear codes, the embedding
is obtained as a syndrome,
$\rm{emb}(\mathbf{c},\mathbf{m})=\mathbf{x}$ with

\[
[SS]:\left\{ \begin{array}{rl}
s(\mathbf{x})&=\mathbf{m},\\
x_i&= c_i \; \mbox{ if $i\in\mathcal{W}$}
\end{array} \right.
\]
\noindent
Also as in the case of linear codes we can ask for the minimum possible number of dry (free) positions required to ensure a solution of $[SS]$, the wet threshold of $\mathcal{C}$. Such a solution exists if and only if
$\pi_{\mathcal{W}}(\mathbf{c}+(\mathbf{0},\mathbf{m}))\in\pi_{\mathcal{W}}(\mathcal{C})$. In that case, if 
$\mathbf{y}\in\mathcal{C}$ verifies $\pi_{\mathcal{W}}(\mathbf{y})=\pi_{\mathcal{W}}(\mathbf{c}+(\mathbf{0},\mathbf{m}))$, then $\mathbf{x}=\mathbf{y}+(\mathbf{0},\mathbf{m})$ is a solution.

\begin{proposition}\label{SolutionS}
If $\dry\ge n-d^{\perp}+1$ then the system $[SS]$ has a solution for all $\mathbf{c}\in\mathbb{F}_2^n$,  $\mathbf{m}\in\mathbb{F}_2^u$ and  $\mathcal{W}\subseteq\{1,\dots,n\}$ with $\#\mathcal{W}=n-\dry$. In this case, $[SS]$ has exactly $2^{\dry-u}$ solutions.
\end{proposition}
\begin{proof}
There exists a solution all $\mathbf{c}\in\mathbb{F}_2^n$,  $\mathbf{m}\in\mathbb{F}_2^u$ and  $\mathcal{W}\subseteq\{1,\dots,n\}$ if and only if $\pi_{\mathcal{W}}(\mathcal{C})=\mathbb{F}_2^{n-\dry}$. The statement follows from Delsarte's theorem and proposition \ref{OAres}.
\end{proof}

Thus the threshold verifies $\tau\le n-d^{\perp}+1$.
If we compare this result with theorem \ref{solutioniff}, we can see a
significant difference: in that case the condition $\dry\ge
n-d^{\perp}+1$ was also necessary. This is due to the existence of dual when the code
$\mathcal{C}$ is linear. If $\mathcal{C}$ is systematic but not
linear, then such dual does not exist and it may happen $\tau< n-d^{\perp}+1$ so that we need less free coordinates than
the required in the linear case. Let us see an example of this situation.

\begin{example}\label{finalex}
The Nadler code has distance distribution $1, 0, 0, 0, 0, 12, 12, 0,
3, 4, 0, 0, 0$ and dual distance distribution
$1,0,0,4,18,36,24,12,21,12,0,0$.  In particular $d^\perp=3$.
When using this code for wet paper purposes, the corresponding system
$[SS]$ has a solution with certainty when the number of locked
coordinates is  $\le d^\perp-1=2$, according to proposition \ref{SolutionS}. A direct inspection shows that
for any $4$ columns of $\mathcal{N}$, every one of the possible $2^4$
vectors of $\mathbb{F}_2^4$ occurs. According to the Bush bound
\cite{ortogonal} this is the maximum possible number of columns for
which this can happen. Then the system $[SS]$ has a solution with
certainty when the number of locked coordinates is at most $4$. Remark
that the minimum distance of a $[12,7]$ linear code is 4, see
\cite{mint}, hence the maximum number of coordinates we can lock   
using linear codes with the same parameters as $\mathcal{N}$ is 3.
\end{example}

The above proposition \ref{SolutionS} and example \ref{finalex} suggest the use of nonlinear systematic codes as wet paper codes. The main drawback of using these codes is in the computational cost of solving $[SS]$. 
Solving a system of boolean equations is a classical and important problem in computational algebra and computer science. There exist several methods available, some of which are very efficient when the number of variables is not too large, see \cite{gao,tesis} and the references therein. Anyway the computational cost of solving $[SS]$ is always greater than that of solving a system of linear equations. In conclusion, the use of nonlinear systematic codes can be an interesting option when the added security gained through a greater number of locked positions offsets the increased computational cost.

\section{Conclusion}\label{Sec:Conclusion}

We have obtained necessary and sufficient
conditions to make sure the embedding process in
the wet paper context. These conditions depend on the dual
distance of the involved code. We also gave a
sufficient condition in the general case of systematic 
codes and provided the exact number of solutions. Finally, we showed that  
systematic codes can be good candidates in the design of wet paper stegoschemes.


\end{document}